\documentclass[aps,prd,twocolumn]{revtex4-1}
\usepackage{enumerate}
\usepackage{amsmath, amsthm, amssymb,commath,braket}
\usepackage{color,calc,graphicx,subfig}

\def\be{\begin{equation}} 
\def\ee{\end{equation}}
\def\bea{\begin{eqnarray}}
\def\eea{\end{eqnarray}}
\def\bma{\begin{mathletters}}
\def\ema{\end{mathletters}}
\def\P{{\cal P}}
\def\0{\overline{0}}
\def\tr{\mbox{tr}}

\def\q0{\underline{0}}

\def\H{{\cal H}}

\def\P{{\mathbb P}}
\def\S{{\mathbb S}}
\def\C{\mathbb{C}}
\def\id{{\mathbb I}}

\def\H{{\cal H}}
\def\U{\mathfrak{U}}

\def\N{\mathbb{N}}

\def\tr{\mbox{tr}}
\def\one{\leavevmode\hbox{\small1\normalsize\kern-.33em1}}

\def\bra#1{\langle#1|} \def\ket#1{|#1\rangle}
\def\braket#1#2{\langle#1|#2\rangle}

\def\proj#1{\ket{#1}\!\bra{#1}}

\newtheorem{theo}{Theorem}

\newtheorem{defin}[theo]{Definition}

\newtheorem{lemma}[theo]{Lemma}

\def\id{{\mathbb I}}
\def\tr{\mbox{tr}}

\begin{document}
\title{Composing decoherence functionals}
\author{Paul Bo\"{e}s,$^1$ and Miguel Navascu\'es$^2$}
\affiliation{$^1$Dahlem Center for Complex Quantum Systems, Freie Universit\"{a}t Berlin, D-14195 Berlin, Germany\\$^2$Institute for Quantum Optics and Quantum Information (IQOQI) Vienna, Austrian Academy of Sciences, Boltzmanngasse 3, 1090 Vienna, Austria}

\begin{abstract}
Quantum Measure Theory (QMT) is a generalization of quantum theory where physical predictions are computed from a matrix known as \emph{decoherence functional} (DF). Previous works have noted that, in its original formulation, QMT exhibits a problem with composability, since the composition of two decoherence functionals is, in general, not a valid decoherence functional. This does not occur when the DFs in question happen to be positive semidefinite (a condition known as strong positivity). In this paper, we study the concept of composability of DFs and its consequences for QMT. Firstly, we show that the problem of composability is much deeper than originally envisaged, since, for any $n$, there exists a DF that can co-exist with $n-1$ copies of itself, but not with $n$. Secondly, we prove that the set of strongly positive DFs cannot be enlarged while remaining closed under composition. Furthermore, any closed set of DFs containing all quantum DFs can only contain strongly positive DFs.

\end{abstract}
\date{}
\maketitle
\section{Introduction}

Despite the continuous efforts of numerous scientists, reconciling General Relativity with Quantum Theory remains one of the most important open problems in Physics. The framework of General Relativity suggests that one promising approach to such a unification will be by means of a reformulation of quantum theory in terms of \emph{histories} rather than \emph{states}. Following this idea, Hartle \cite{hartle1,hartle2}, and, independently, Sorkin \cite{sorkin1,sorkin2,sorkin3}, have proposed a history-based framework, which can accommodate both standard quantum mechanics as well as physical theories beyond the quantum formalism. In this framework, dubbed Generalized Quantum Mechanics or Quantum Measure Theory, a fundamental role is played by the so-called decoherence functional (DF). This object not only determines which questions can be answered regarding the history of a closed quantum system, but it also assigns probabilities to each possible answer. Previous works on Quantum Measure Theory (QMT) have discussed the interpretation of the DF \cite{sorkin3} (still an active topic in standard quantum theory, see \cite{open}), proven the equivalence between the DF and operator-based formulation of quantum mechanics \cite{hilbert} and established a link with the operational axiomatization of quantum theory \cite{henson}. In this last regard, the authors of \cite{henson} argue that, when we try to model nonlocality experiments within the framework of DFs, the almost quantum set of correlations naturally appears \cite{almost}. The almost quantum set is a set of nonlocal correlations that, despite being strictly larger than the quantum one, seems to satisfy all information-theoretic principles originally conceived to single out quantum nonlocality \cite{prin1, prin2, prin3, prin4, prin5, prin6}. This feature, together with its unexpected role in QMT led the authors of \cite{almost} to conjecture that the almost quantum set represents the set of correlations of a yet-to-be-discovered consistent physical theory, arguably more plausible than quantum mechanics itself.

There are, though, some caveats. Some of the above results are derived under the assumption that, when viewed as matrices, DFs must be positive semidefinite, a condition known as \emph{strong positivity} (SP). As noted in \cite{random}, the composition of two DFs is not necessarily a DF. The set of strongly positive DFs does not present such patologies, and so the authors of \cite{random,hilbert} suggest that the axiom of SP must be adopted on this basis.

In this paper we study the composability properties of DFs, in an attempt to put this last suggestion on a firmer ground. By postulating a very natural rule for the composition of DFs corresponding to independent systems, we will derive two results on the composability of DFs. Firstly, we will argue that the problem of the composability of DFs is much deeper than originally noted, by showing that, for any $n$, there exists a DF such that $n$ copies of it can co-exist simultaneously, but not $n+1$ copies. Secondly, we will prove that no theory that is closed under composition and that extends the set of DFs of quantum theory can contain DFs which are not SP. Consequently, by \cite{henson}, the set of nonlocal correlations of such a theory must be contained in the almost quantum set.

The structure of this paper is as follows: In Section \ref{QMT}, we will introduce QMT and explain how it relates to the decoherence histories interpretation of quantum theory. We will explain how to model Bell experiments in this framework and review the work of \cite{henson}. We will also formulate a product rule for the composition of DFs describing independent physical systems. In Section \ref{results}, we will explore the consequences of the product rule. We will present our conclusions in Section \ref{conclusion}.


\section{Quantum Measure Theory}
\label{QMT}

The basic idea behind Quantum Measure Theory, or Generalized Quantum Mechanics, for that matter, is to provide a description of the world in terms of \emph{histories}. A history is a classical description of the system under consideration for a given period of time, finite or infinite. If we are trying to describe a system of $N$ particles, then a history will be given by $N$ classical trajectories. If we are working with a field theory, then a history will correspond to the spatial configuration of the field as a function of time. In either case, QMT tries to provide a way to describe the world through classical histories by extending the notion of probability theory (which, in the light of recent experiments \cite{bell_exp1, bell_exp2, bell_exp3}, is clearly not rich enough to model our universe). QMT does so via an object dubbed \emph{decoherence functional} (DF).

Call $\Omega$ the set of possible histories of the system; and $\U$, the set of all subsets of $\Omega$. A DF is a function $D:\U\times\U\to \C$, with the following properties:

\begin{enumerate}
\item{\emph{Hermiticity.}}
$D(A|B)=D(B|A)^*$.
\item{\emph{Linearity.}}
$D(A\cup B|C)=D(A|C)+D(B|C)$, if $A\cap B=\emptyset$.
\item{\emph{Positivity}}
$D(A|A)\geq 0$.
\item{\emph{Normalization.}}
$D(\Omega,\Omega)=1$.
\item{\emph{Weak decoherence.}}
Let $\{A_k\}_{k}$ be a partition of $\Omega$. If there exists a feasible experiment to determine which $A_k$ describes the system, then 

\be
\mbox{Re}\{D(A_k|A_j)\}=P(A_k)\delta_{kj},
\label{decoherence}
\ee

\noindent where $P(A_i)$ is the probability that the history of the system belongs to $A_i$.
\end{enumerate}

The delta appearing in eq. (\ref{decoherence}) is required by consistency. Indeed, one would imagine that, if there existed an experiment capable of distinguishing between the elements of the partition $\{A_i\}_i$, then, for $k\not=j$, there should be an experiment distinguishing the elements of the coarse-grained partition $\{A_i:i\not=j,k\}\cup \{A_k\cup A_j\}$, with $D(A_k\cup A_j|A_k\cup A_j)=P(A_k\cup A_j)=P(A_k)+P(A_j)=D(A_k|A_k)+D(A_j|A_j)$. This last condition and the axioms of linearity and hermiticity imply that $\mbox{Re}\{D(A_k|A_j)\}=0$. 

The difference between QMT and Generalized Quantum Mechanics is that the latter also posits the converse of the Decoherence Axiom. Namely, that, for any partition $\{A_i\}_i$ of $\Omega$ satisfying $\mbox{Re}\{D(A_k|A_j)\}=0$ for $k\not=j$, there exists a feasible experiment to determine which partition the system is in, and the probabilities of each outcome are given by eq. (\ref{decoherence}) \cite{DF}. In \cite{DF} it is observed that, in many situations of physical interest, not only the real part of $D(A_k|A_j)$ vanishes, but also the imaginary part. This has led some authors \cite{DF,hilbert,henson} to postulate the stronger axiom

\begin{enumerate}
\setcounter{enumi}{4}
\item{\emph{Strong decoherence.}}
Let $\{A_k\}_{k}$ be a partition of $\Omega$. If there exists a feasible experiment to determine which $A_k$ describes the system, then 

\be
D(A_k|A_j)=P(A_k)\delta_{kj},
\label{decoherence}
\ee

\noindent where $P(A_i)$ is the probability that the history of the system belongs to $A_i$.
\end{enumerate}

Now, to clarify the connection between DFs and reality we will next model a Bell experiment within the DF framework as in \cite{henson}. Let two parties, call them Alice and Bob, conduct separate experiments labeled by $x\in\{1,..,m\}$ and $y\in\{1,..,m\}$, thus obtaining outcomes $a\in\{1,...,d\}$, and $b\in\{1,...,d\}$. They observe that the statistics of their experiments are governed by the list of probabilities $\{P(a,b|x,y)\}_{x,y,a,b}$. For simplicity, we will call this list \emph{behavior} and just denote it by $P(a,b|x,y)$.

To model this experiment within the framework of DFs, call $a_x$ ($b_y$) the property of the system revealed when Alice (Bob) conducts the experiment $x$ ($y$). We regard as a history any list $\omega\equiv(a_1,...,a_n,b_1,...,b_n)\in\{1,...,d\}^{2n}$ of all the properties of the system. Now, for fixed $x,y$, consider the partition $\{A_{x,y,a,b}\}_{a,b}$ of $\U$ given by the sets $A_{x,y,a,b}=\{\omega:\omega(a_x)=a,\omega(b_y)=b\}$. Which element of this partition the history of our system falls into is a question that can be answered just by forcing Alice and Bob to conduct measurements $x,y$. Hence the elements of the partition must decohere, i.e., $D(A_{x,y,a,b}|A_{x,y,a',b'})=P(a,b|x,y)\delta_{a,a'}\delta_{b,b'}$ (assuming strong decoherence). Similarly, the act of Alice conducting an experiment $x$ and asking Bob to conduct experiment $y$ depending on her outcome $a$ (and viceversa) determines further decohering partitions. Note that, given $P(a,b|x,y)$, this reasoning does not allow us to fill all the entries of the DF, but to establish affine relations between them.

If we wanted to model this experiment within standard quantum mechanics, we would need to assume that there exists a Hilbert space $\H$, a normalized quantum state $\rho\in B(\H)$ and projector operators $E(x,a),F(y,b)\in B(\H)$ satisfying $\sum_{a}E(x,a)=\sum_{b}F(y,b)=\id$, $[E(x,a),F(y,b)]=0$ such that $P(a,b|x,y)=\tr(\rho E(x,a) F(b,y))$. Then, a possible \emph{quantum DF} to describe this system follows from linearity from the following prescription:

\begin{align}
&D((a_1,...,a_n,b_1,...,b_n)|(a'_1,...,a'_n,b'_1,...,b'_n))=\nonumber\\
&\tr\left\{E(\vec{a})F(\vec{b})\rho F(\vec{b}')^\dagger E(\vec{a}')^\dagger\right\}, \label{eq:quantum}
\end{align}

\noindent 
where $E(\vec{a})\equiv\prod_{x=1}^nE(x,a_k)$, $F(\vec{b})\equiv\prod_{y=1}^nF(y,b_k)$.

Sometimes it is more convenient to represent the decoherence functional as an operator. Given the set of all histories $\Omega$, one does so by defining a Hilbert space by assigning to each $a\in\Omega$ an element $\ket{a}$ of an orthonormal basis $\{\ket{b}:b\in\Omega\}$. Next, to any $A\in \U$ we associate the vector $\ket{A}=\sum_{a\in A}\ket{a}$. Finally, for any DF, we can build the matrix $D=\sum_{a,b\in \Omega}D(\{a\}|\{b\})\ket{a}\bra{b}$. By the axiom of linearity we then have that $D(A|B)=\bra{A}D\ket{B}$. 

In this matrix representation, we can thus drop linearity and the remaining axioms can be expressed as:

\begin{enumerate}
\item{\emph{Hermiticity.}}
$D=D^\dagger$.
\item{\emph{Positivity.}}
$\bra{V}D\ket{V}\geq 0$, for all vectors $\ket{V}\in\{0,1\}^{|\Omega|}$.
\item{\emph{Normalization.}}
$\bra{\Omega}D\ket{\Omega}=1$.
\item{\emph{Weak (Strong) decoherence.}}
Let $\{A_k\}_{k}$ be a partition of $\Omega$. If there exists a feasible experiment to determine which $A_k$ describes the system, then 

\be
\mbox{Re}\{\bra{A_j}D\ket{A_k}\}\mbox{ }\left(\bra{A_j}D\ket{A_k}\right)=P(A_k)\delta_{kj},
\label{decoherence}
\ee

\noindent where $P(A_i)$ is the probability that the history of the system belongs to $A_i$.
\end{enumerate}

\noindent A DF is \emph{strongly positive} if its matrix representation is positive semi-definite, i.e. if all its eigenvalues are non-negative. As illustrated in the example \eqref{eq:quantum}, all quantum DFs are strongly positive \cite{random}. However, the set of all behaviors $P(a,b|x,y)$ admitting a strongly positive DF is not the quantum set, but the set of \emph{almost quantum correlations} \cite{henson}, defined in \cite{almost}.

Now, consider two physical systems with history sets $\Omega_1,\Omega_2$ and decoherence functionals $D_1:\U_1\times\U_1\to \C$, $D_2:\U_2\times\U_2\to \C$. Since the decoherence functional is the basic object of QMT, there should exist a prescription that assigns a joint DF $D_{12}: \U_{12} \times \U_{12}\to \C$ to the joint system $1-2$. Here $\U_{12}$ denotes the power set of the history space $\Omega_{12} = \Omega_1 \times \Omega_2$. 

Taking inspiration from the quantum case, we suggest the following composition rule, assumed implicitly in \cite{random}:
\begin{defin}{\emph{Product rule}}

Let $D_1, D_2$ be DFs describing independent systems, with history sets $\Omega_1,\Omega_2$. Then the joint DF $D_{12}:\U_{12}\times\U_{12}\to \C$ describing systems $1$ and $2$ is completely determined by linearity and the relations:

\be
D_{12}(A_1\times A_2|B_1 \times B_2)=D_1(A_1|B_1)D_2(A_2|B_2),
\ee

\noindent for all $A_1,B_1\in \U_1,A_2,B_2\in \U_2$.
\end{defin}

\noindent Note that, in operator representation, the product rule reduces to $D_{12}=D_1\otimes D_2$. 

As noted by Di\'osi, under the product rule the composition of any two DFs violating strong decoherence results in a matrix that violates the axiom of \emph{weak} decoherence \cite{diosi}. Thus, unless the system under study is literally \emph{singular}, we can safely assume that measurable partitions of $\Omega$ satisfy strong decoherence. Di\'osi's composition argument hence settled the controversy on whether the axiom of weak or strong decoherence was to be preferred.

In \cite{random,hilbert}, the authors further consider the possibility that the composition of two DFs via the product rule give rise to a matrix violating the positivity axiom. That such pathological examples do exist follows from the recent results of Henson \cite{henson2}, which imply that there exists a behavior $P(a,b|x,y)$ (specifically the Popescu-Rohrlich box \cite{pr}), such that one copy admits a DFs but not two independent copies.

However, as pointed out in \cite{random,hilbert,henson}, strongly positive DFs do not exhibit this behavior. Indeed, for $D_1,D_2\geq 0$, it follows straightforwardly that $D_{12}=D_1\otimes D_2\geq0$. In other words: the set $\S$ of strongly positive DFs is \emph{closed under composition}.

A set ${\cal S}$ of DFs is closed under composition iff, for any $D, D'\in {\cal S}$, $D\otimes D'\in {\cal S}$. Intuitively, any closed system described by such a set can in principle co-exist with any other closed systems subject to the same constraint. This is very desirable from a physical point of view, as it allows us to model a closed system independently of the rest of the universe. 

By adopting the product rule and requiring closure under composition, the authors of \cite{random, hilbert, henson} motivate the axiom of strong positivity. However, this is only true to some extent, because $\mathbb{S}$ is not the only closed set within the set of DFs and it is not clear whether all other closed sets are subsets of $\mathbb{S}$. For all we know, we could weaken the SP condition and yet end up with a closed set of DFs, possibly describing nonlocal correlations beyond the almost quantum set.

Moreover, it may be argued that the requirement of closure under composition is unnecessarily strong, given that the universe may just hold a finite number of independent systems. In the following section both of these questions are addressed.

\section{Consequences of the product rule}
\label{results}

Is requiring that the set of admissible DFs of a physical theory is closed under composition an unnecessarily strong assumption? One possible argument \emph{in favour of} this requirement follows from the observation in \cite{henson2} that there exists a DF $D$ such that $D^{\otimes 2}$ is not a DF. Hence, in order to postulate that a given system is described by $D$, we would need to be certain that $D$ does not describe any other system \emph{in the whole universe}. By admitting the feasibility of $D$, we would be thus renouncing to the possibility of modeling closed systems independently of the rest of the world.

This argument, however, relies on the intuition that, for cosmic scales, $2$ is a very small number. Thus, an argument \emph{against} the requirement of closure may counter that such a reasoning cannot hold for astronomic numbers of copies: what if a very high number of copies of $D$, say $10^{200}$ -many more than what the Universe can actually accommodate-, were compatible? Would it be then reasonable to conclude that $D$ cannot be realized just because $10^{200}+1$ copies cannot theoretically co-exist? The next result shows that the above is a valid concern:

\begin{lemma}
Let $n\in \N$. Then, there exists a decoherence functional $D$ such that $D^{\otimes n}$ is a decoherence functional, but $D^{\otimes n+1}$ is not.
\end{lemma}

\begin{proof}
Consider the $4\times 4$ matrix

\be
D\equiv \frac{1}{2}\epsilon A\otimes \proj{0}+\frac{1}{2}(\id-\epsilon A)\otimes \proj{1},
\ee
\noindent with 

\be
A=\left(\begin{array}{cc}1&\lambda\\\lambda&1\end{array}\right),
\ee

\noindent for some $\lambda>1$. 

As long as $\epsilon\leq \frac{1}{1+\lambda}$, this can be interpreted as a weakly positive decoherence functional for two non-compatible measurements with dichotomic outcomes. Indeed, denote each history by a vector $(a,b)\in\{0,1\}^2$ and define $D(a,b|a',b')\equiv\bra{a}\bra{b} D\ket{a'}\ket{b'}$. Then it can be verified that the partitions $\{A_a\}_a$ and $\{B_b\}_b$, with $A^a=\{(a,b):b=0,1\}$, $B^b=\{(a,b):a=0,1\}$ strongly decohere. Moreover, $\bra{C}D\ket{C}\geq 0$ for any subset $C$ of $\Omega=\{(a,b):a,b=0,1\}$.

Now, consider the element of $\U_{1,...,n+1}$ given by the vector $\ket{V}=\ket{0,0}^{\otimes n}\ket{0}\ket{1}+\ket{1,0}^{\otimes n}\ket{1}\ket{1}$. One can verify that $\bra{V}D^{\otimes n+1}\ket{V}=\frac{1}{2^{n}}\epsilon^n[1-\epsilon(1+\lambda^{n+1})]$. Hence, for $\epsilon> 1/(\lambda^{n+1}+1)$, $\bra{V}D^{\otimes n+1}\ket{V}<0$, and so $D^{\otimes n+1}$ is not a DF. In particular, taking $\epsilon=1/(\lambda^{n+1/2}+1)$ we make sure that the positivity axiom is violated.

It just rests to see that $\lambda$ can be chosen such that $D^{\otimes n}$ is a decoherence functional. Note that, due to the block-diagonal structure of $D$, it is enough to show that $\bra{W}A^{\otimes n_1}\otimes B^{\otimes n_2}\ket{W}\geq 0$ for any $\ket{W}\in\{0,1\}^{2^n}, n_1+n_2=n$, with $B=\id-\epsilon A$. 

Define $\ket{\tilde{W}}=\ket{W}/\|\ket{W}\|$. Then,

\begin{align}
&\bra{\tilde{W}}A^{\otimes n_1}\otimes B^{\otimes n_2}\ket{\tilde{W}}\nonumber\\
&=\bra{\tilde{W}}A^{\otimes n_1}\otimes \id\ket{\tilde{W}}+\bra{\tilde{W}}A^{\otimes n_1}\otimes (B^{\otimes n_2}-\id)\ket{\tilde{W}}\nonumber\\
&\geq 1-\|A^{\otimes n_1}\otimes (B^{\otimes n_2}-\id)\|_\infty\nonumber\\
&=1-(1+\lambda)^{n_1}\left\{\left(1-\frac{\lambda+1}{\lambda^{n+1/2}+1}\right)^{n_1}-1\right\}.
\label{positivity}
\end{align}

\noindent Here the inequality $\bra{\tilde{w}}A^{\otimes n_1}\otimes \id\ket{\tilde{w}}\geq 1$ follows from the fact that $A=\id+\tilde{A}$, with all the entries of $\tilde{A}$, $\ket{\tilde{W}}$ being non-negative.

We claim that, for any $n$ and $\lambda$ sufficiently high, the right hand side of eq. (\ref{positivity}) is arbitrarily close to $1$ for all $n_1+n_2=n$, $n_2\geq 1$. To see this, one just needs to take the limit $\lambda\to \infty$ in that expression. Hence $\bra{\tilde{W}}A^{\otimes n_1}\otimes B^{\otimes n_2}\ket{\tilde{W}}\geq 0$ for all $n_1+n_2=n$ (the case $n_2=0$ is trivial) and so $D^{\otimes n}$ is a DF.

\end{proof}

In the light of this result, whether the requirement of closure under composition is unnecessarily strong or not is a question that cannot be settled via arguments involving small numbers of independent systems. 

\vspace{10pt}

Suppose now that, despite the previous lemma, we adopt as a physical principle that DFs must be composable arbitrarily many times. This would imply that the set ${\cal S}$ of DFs deemed physical must be closed under the product rule. That is, if $D_1,D_2\in {\cal S}$, then $D_1\otimes D_2\in {\cal S}$. In order to meet this requirement, it has been proposed \cite{henson,hilbert} that ${\cal S}$ should be contained in the set $\S$ of strongly positive DFs.

In the following we will show that, under the product rule, $\S$ is maximal, i.e., it cannot be enlarged without losing the property of being closed under composition. Actually, we will prove an even stronger result: namely, that any set of DFs containing the quantum set cannot be both closed under composition and have a non-SP element. Any theory aiming at reproducing quantum predictions for a certain family of experiments (described below) can just thus contain SP DFs.

\begin{lemma}
\label{maximal}
Let $D$ be a decoherence functional violating strong positivity, i.e., there exists $\ket{v}$ with $\bra{v}D\ket{v}<0$. Then, there exists a quantum decoherence functional $D'$ describing an experience with two measurement settings such that $D\otimes D'$ is not a decoherence functional.
\end{lemma}

\begin{proof}

\noindent In order to prove the theorem, we need to define the following family of DFs:

\begin{defin}
Let $\ket{v}\in \C^m$ be a normalized vector, i.e., $\braket{v}{v}=1$. Then, $D_{\ket{v}}\in B(\C^m\otimes\C^2)$ will denote the quantum decoherence functional given by

\be
D_{\ket{v}}(a,b|a',b')\equiv \bra{v}E_{a'}F_{b'}F_{b}E_{a}\ket{v},
\ee,

\noindent where $E_a=\proj{a}$, for $a=0,...,m-1$ are 1-rank projectors onto the computational basis of $\C^m$, and $F_0,F_1$ correspond to the projectors $F_0=\frac{1}{m}\sum_{j,k}\ket{j}\bra{k}$, $F_1=\id_m-F_0$.
\end{defin}

By straightforward calculation, one can verify that

\begin{align}
&D_{\ket{v}}=\frac{1}{m}\proj{v}\otimes \proj{0} +\nonumber\\
&\left(\sum_{a=0}^{m-1}|\braket{v}{a}|^2\proj{a}-\frac{1}{m}\proj{v}\right)\otimes \proj{1}.
\end{align}

\noindent Defining $\ket{w}\equiv\sum_{a=0}^{m-1}\ket{a}_A\ket{a,0}_B$, it hence follows that

\be
\tr_B\{\proj{w}_{AB}(\id_A\otimes D_{\ket{v}})\}=\frac{1}{m}\proj{v^*}.
\ee

\noindent Note that, expressed in the computational basis, $\ket{w}$ has just zeros and ones. Therefore it can be identified with a subset of the joint set of histories $\Omega_{AB}$.

Now we are ready to prove the maximality result: let $D$ be a weakly positive DF, with a set of histories $\Omega$ with cardinality $m$ and such that $\bra{v}D\ket{v}<0$, for some normalized vector $\ket{v}\in\C^m$. Then,

\be
\bra{w}D\otimes D_{v^*}\ket{w}=\frac{1}{m} \bra{v}D\ket{v}<0.
\ee 
\noindent The positivity axiom is thus violated by the composition of $D$ with $D_{v^*}$.
\end{proof}

\vspace{10pt}
The last lemma gives further support to the adoption of strong positivity as an axiom for the theory of DFs. However, it does not close the matter completely: that $\S$ is a maximal set closed under composition does not necessarily mean that it is the \emph{only} maximal set. In principle, there could exist other sets of DFs closed under composition containing elements with negative eigenvalues. Each of these other maximal sets would give rise to a consistent theory of DFs (although unable, in the light of Lemma \ref{maximal}, to reproduce all the predictions of quantum mechanics).

The main obstacle in showing the existence of other maximal sets is the difficulty of identifying $s\times s$ matrices $D\not\geq 0$ all of whose $n$-tensor products $S=D^{\otimes n}$ satisfy 

\be
\bra{u}S\ket{u}\geq 0,
\label{one-positivity}
\ee

\noindent for all vectors $\ket{u}\in\{0,1\}^{sn}$. An obvious choice is taking $D$ to have only non-negative entries. Actually, the set $\P$ of \emph{hermitian matrices} with non-negative coefficients can be shown to be maximal under composition, in the sense that they and their tensor products satisfy relation (\ref{one-positivity}), and, for any $s\times s$ matrix $M\not\in \P$, there exists a $2\times 2$ matrix $D$ in $\P$ such that $M\otimes D$ violates (\ref{one-positivity}). As we will next see, though, this choice does not lead to very interesting DFs. 

Consider a set $\Omega$ of histories labeled by the values of $n$ properties, i.e., each history is of the form $a\equiv(a_1,a_2,...,a_n)$. Now let $D$ be the matrix representation of a DF with non-negative entries, and suppose that there exists a non-zero off-diagonal element $D(b|c)\not=0$, with $b\not=c$. Since $b\not=c$, let $k$ be any index such that $b_k\not= c_k$ and consider the partition of $\Omega$ given by the sets $A_{a}\equiv \{(a_1,...,a_{k-1},a,a_{k+1},...,a_n):a_1,...,a_{k-1},a_{k+1},...,a_n\}$. Then, the fact that all entries of $D$ are non-negative implies that $D(A_{b_k}|A_{c_k})\not=0$. That is, this partition does not decohere, and consequently there exists no physical procedure to determine the value of property $a_k$. That we cannot measure one of the properties which we used to define our set of histories is clearly nonsensical (or, at the very least, extremely undesirable), so we must conclude that $D$ must be a diagonal matrix. Hence $D$ must be diagonal or classical and thus strongly positive.

\section{Conclusion}
\label{conclusion}
In this paper we have proposed a natural rule for the composition of DFs and studied its consequences within the axiomatization of QMT. We have shown that this extra rule predicts the existence of families of DFs which, despite being $n$-fold compatible, cannot co-exist with $n+1$ copies of themselves. Also, we have proven that, under this composition law, the set $\S$ of strongly positive DFs is maximal, in the sense that it cannot be enlarged without losing the property of being closed under composition.

It remains to find out whether $\S$ is actually the only maximal set, or, on the contrary, there exist other sets not contained in $\S$ but nonetheless closed under composition. Alternatively, one could look for a new composition rule for which $\S$ is not maximal. We suspect that both problems are similarly involved.

\vspace{10pt}
\noindent\emph{Acknowledgements}
M.N. acknowledges the FQXi grant `Towards an almost quantum physical theory' and useful discussions with A. Ac\'in and J. Henson.
P.B. would like to thank F. Dowker and J. Henson for supervising parts of this research as well as the Quantum Information group at University Aut\'{o}noma Barcelona for their hospitality during a research visit for this project.


\begin{thebibliography}{99}
\bibitem{hartle1}
J. B. Hartle, \emph{The quantum mechanics of cosmology Dec 1989}, Lectures at Winter School on Quantum Cosmology and Baby Universes, Jerusalem, Israel, Dec 27, 1989 - Jan 4, 1990.
\bibitem{hartle2}
J. B. Hartle, \emph{Spacetime quantum mechanics and the quantum mechanics of spacetime}, in Proceedings of the Les Houches Summer School on Gravitation and Quantizations, Les Houches, France, 6 Jul - 1 Aug 1992, eds J. Zinn-Justin, and B. Julia, North-Holland (1995), arXiv:gr-qc/9304006.
\bibitem{sorkin1}
R. D. Sorkin, \emph{Quantum Mechanics as Quantum Measure Theory}, Mod. Phys. Lett. A9 3119-3128 (1994).
\bibitem{sorkin2}
R. D. Sorkin, \emph{Quantum Dynamics without the Wave Function}, J. Phys. A: Math. Theor. 40 3207-3221 (2007).
\bibitem{sorkin3}
R. D. Sorkin, \emph{Quantum Measure Theory and Its Interpretation}, in Quantum Classical Correspondence: Proceedings of the 4th Drexel Symposium on Quantum Nonintegrability, pages 229-251, International Press, Cambridge Mass. 1997, D.H. Feng and B-L Hu (editors).
\bibitem{open}
R. de Melo e Souza, F. Impens and P. A. Maia Neto, arXiv:1603.02522.
\bibitem{random}
X. Martin, D. O’Connor and R.D. Sorkin, \emph{Random walk in generalized quantum theory}, Phys. Rev. D {\bf 71}, 024029 (2005).
\bibitem{hilbert}
F. Dowker, S. Johnston and R. D. Sorkin,  J. Phys. A {\bf 43}, 275302 (2010).
\bibitem{henson}
F. Dowker, J. Henson and P. Wallden, New J. Phys. {\bf 16}, 033033 (2014).
\bibitem{henson2}
J. Henson, Phys. Rev. Lett. {\bf 114}, 220403 (2015)
\bibitem{almost}
M. Navascu\'es, Y. Guryanova, M. J. Hoban and A. Ac\'in, Nat. Comm. {\bf 6}, 6288 (2015).
\bibitem{prin1}
S. Popescu and D. Rohrlich, Foundations of Physics {\bf24} (3): 379�385 (1994).
\bibitem{prin2}
G. Brassard, H. Buhrman, N. Linden, A. A. Methot, A. Tapp and F. Unger, F., Phys. Rev. Lett., 96 250401, (2006).
\bibitem{prin3}
N. Linden, S. Popescu, A. J. Short, and A. Winter, Phys. Rev. Lett. 99, 180502 (2007).
\bibitem{prin4}
M. Pawlowski, T. Paterek, D. Kaszlikowski, V. Scarani, A. Winter, and M. Zukowski, Nature 461, 1101 (2009).
\bibitem{prin5}
M. Navascu\'es and H. Wunderlich, Proc. Royal Soc. A 466:881-890 (2009).
\bibitem{prin6}
T. Fritz, A. B. Sainz, R. Augusiak, J. B. Brask, R. Chaves, A. Leverrier and A. Ac\'in, arXiv:1210.3018.
\bibitem{pr}
S. Popescu, D. Rohrlich, Found. Phys. \textbf{3} 24, 379-385, (1994)
\bibitem{griffiths}
R. Griffiths, J. Stat. Phys. {\bf 36}, 219 (1984).
\bibitem{omnes}
R. Omnes, Ann. Phys. {\bf 201}, 354 (1990).
\bibitem{gellmann}
M. Gellmann and J. B. Hartle, in \emph{Complexity, Entropy and the Physics of Information}, eidted by W. Zurek, Santa Fe Institute Studies in the Sciences of Complexity Vol. VIII (Addison-Wesley, Reading, 1990).
\bibitem{DF}
H. F. Dowker and J. J. Halliwell, Phys. Rev. D {\bf 46}, 1580 (1992).
\bibitem{bell_exp1}
B.Hensen \emph{et al.}, Nature {\bf 526}, 682–686 (2015).
\bibitem{bell_exp2}
M. Giustina \emph{et al.}, Phys. Rev. Lett. {\bf 115}, 250401 (2015).
\bibitem{bell_exp3}
L. K. Shalm \emph{et al.}, Phys. Rev. Lett. {\bf 115}, 250402 (2015).
\bibitem{diosi}
L. Di\'osi, Phys. Rev. Lett. {\bf 92}, 170401 (2004).

\end{thebibliography}
\end{document}